\documentclass[11pt]{article}

\usepackage{amsmath}
\usepackage{amsthm}
\usepackage{comment}
\usepackage{amsmath}
\usepackage{amsfonts} 
\usepackage{algorithm2e, algorithmicx, algpseudocode} 
\usepackage{bbm} 
\usepackage{mathdots}
\RestyleAlgo{boxed}
\usepackage{thmtools} 
\usepackage{thm-restate}
\usepackage{hyperref}
\usepackage{times}
\usepackage[normalem]{ulem}
\usepackage{url}

\newtheorem*{theorem*}{Theorem}

\newtheorem{corollary}{Corollary}[section]
\newtheorem{lemma}{Lemma}

\usepackage{graphicx}
\usepackage{xspace}
\usepackage[shortlabels]{enumitem} 
\usepackage[skins]{tcolorbox}
\usepackage{subcaption}
\usepackage{appendix}
\usepackage{wrapfig}%
\usepackage{authblk} 

\newtheorem{myTheorem}{Theorem}
\newcommand{\uar}{u.a.r.\xspace}

\newcommand{\numIns}{\mathcal{I}\xspace}
\newcommand{\numQue}{\mathcal{Q}\xspace}
\newcommand{\numDel}{\mathcal{D}\xspace}
\newcommand{\listmax}{\ell_{M}\xspace}
\newcommand{\listave}{\ell_{\mbox{\tiny ave}}\xspace}

\newcommand{\InsertCost}{\mathcal{A}^{\mbox{\tiny ins}}\xspace}

\newcommand{\htablesize}{t\xspace}
\newcommand{\defn}[1]{\textbf{\emph{#1}}}

\newcommand{\attackindices}{s\xspace}
\newcommand{\AdvTotal}{{\mathcal{B}}\xspace}

\newcommand{\AlgoName}{\textsc{Depth Charge}\xspace}

\newif\ifcomments
\commentstrue

\title{Defending Hash Tables from Subterfuge with Depth Charge\thanks{This work is supported by NSF awards CNS-2210299, CNS-2210300, and CCF-2144410.}}

\author[1]{Trisha Chakraborty}
\author[2]{Jared Saia}
\author[3]{Maxwell Young}

\affil[1,3]{\small Department of Computer Science and Engineering, Mississippi State University, MS, USA\hspace{3cm} \texttt{tc2006@msstate.edu}, \texttt{myoung@cse.msstate.edu}\medskip}

\affil[2]{\small Department of Computer Science, University of New Mexico, NM, USA\hspace{5cm} 
\texttt{saia@cs.unm.edu}\medskip}

\date{}
\begin{document}
\maketitle

\vspace{-1.3cm}

\begin{abstract}
We consider the problem of defending a hash table against a Byzantine attacker that is trying to degrade the performance of query, insertion
and deletion operations. Our defense makes use of resource
burning (RB)---the verifiable expenditure of network resources---where the issuer of a request incurs some RB cost. Our algorithm,
\AlgoName, charges RB costs for operations based on the
depth of the appropriate object in the list that the object hashes to in the table. By appropriately setting the RB costs, our algorithm mitigates the impact of an attacker on the hash table’s performance. In particular, in the presence of a significant attack, our algorithm incurs a cost which is asymptotically less that the attacker’s cost.
\end{abstract}

\section{Introduction}


While hash tables are a popular data structure, their performance  can be significantly degraded if the objects to be stored are chosen adversarially~\cite{Crosby2003AC, Yosef2007RemoteAC, Tobin2012Hashpile}.  In an extreme case, all objects can be hashed to the same  index of the table. Under the common collision-resolution method of chaining, this attack effectively transforms the hash table into a linked list, which leads to a worst-case query time that is linear in the number of objects; this is an example of an  \defn{algorithmic complexity attack (ACA)}~\cite{Yosef2007RemoteAC, Crosby2003AC,cai2009exploiting,sun2011covert,khan2005prevention,ben2014computer}. Many data structures are vulnerable to ACAs, and designing a defense is challenging, since malicious inputs need not be large, or arrive at a high rate, in order to degrade performance; in other words, {\it ACAs are often less costly to launch than they are to defend against}. 

In hash tables, a common defensive measure is to keep the hash function secret (known only to the server) and use stronger (cryptographic) hash functions that are difficult to invert. However, side-channel attacks may allow an adversary to learn the hash function~\cite{Olekšák2022siphash}, and in distributed settings where the hash table may be stored on multiple machines, a single compromised machine may reveal the secret. Similarly, stronger hash functions offer insufficient protection, as an adversary can find objects that hash to the same index through trial and error. These vulnerabilities are discussed in Section~\ref{sec:related-work}, but a fundamental shortcoming of prior defenses is that they do not counteract the cost advantage enjoyed by the attacker.


In this work, we design and analyze a new defense for hash tables that employs \defn{resource burning (RB)}---the verifiable expenditure of a network resource---to reverse this cost asymmetry. Specifically, any user  wishing to access the hash table must pay an RB cost. By setting the amount of RB appropriately, our defense guarantees that the cost to legitimate users grows slowly as a function of an attacker's cost for launching an ACA. In practice, attackers must often pay for the resources needed to launch attacks, such as  renting compromised machines  (see~\cite{franklin:inquiry}). Therefore, the asymptotic advantage given by our approach ultimately translates into a financial edge for the defenders.
\medskip  

\noindent{\bf Our Setting.} We consider the challenging setting where both (1) the number of indices in the hash table; and (2) the hash function are fixed.\footnote{Our approach could be combined with heuristics that dynamically update the number of indices or the hash function.} Our ``fixed'' setting is particularly relevant for many applications in distributed computing. For example, in the client-server setting, changing the hash table size or the hash function can result in down time that negatively impacts quality of service for the clients.  Thus, system administrators often analyze  workload data to set the hash table size appropriately~\cite{ibm:size}. In peer-to-peer systems such as distributed hash tables (DHTs) (e.g.~\cite{6688697,falkner:profiling, stoica_etal:chord,stading2002peer}), memory and disk space is bound by the number of participating machines, and thus resizing is not possible.

Our approach allows for flexibility in setting the appropriate table size. Specifically, we parameterize our results by {\boldmath{$\listmax$}}, which is the maximum number of (legitimate) objects that are hashed to the same index. Intuitively, this parameter is small when the table is appropriately sized (see Section~\ref{sec:main-results} for  discussion).

\subsection{Model}\label{sec:model}
In our setting, there are clients, an adversary, and a server. Note that all clients are ``good''; we do not refer to ``bad'' clients, since the adversary incarnates them. Our server may represent multiple real-world servers. We now describe the key aspects of our model.
\medskip

\noindent{\bf Hash Table.} The server holds a hash table that services insertions, queries, and deletions of objects by request from the clients and the adversary. An insertion by a client is a \defn{good insertion} and the corresponding object is said to be a \defn{good object}, which is placed at an index selected uniformly at random (\uar).\footnote{We can view this as performing a hash function evaluation on the object or, more commonly, an identifier/key value/name of the object. The output is an index in the hash table whose first evaluation is an index of selected uniformly at random, while subsequent evaluations of this object always map to the same index (i.e., our hash function obeys the random oracle assumption~\cite{koblitz2015random}). } Otherwise, the insertion is a \defn{bad insertion} and the corresponding object is a \defn{bad object}; in this case, the adversary selects the index where the object is inserted. Good insertions cannot be distinguished from bad insertions; similarly, good objects cannot be distinguished from bad objects.

A collision occurs when two or more objects are inserted at the same index of the hash table, and this is resolved via the popular method of chaining. That is, the objects involved in the collision form  a \defn{list}, where the head of the list (\defn{HoL}) is located in the index of the hash table, with subsequent objects added to the tail of the list  (\defn{ToL}) in the order that they are inserted.  The \defn{length} of a list at index $i$ is the number of objects stored at index $i$. The \defn{depth} of an object is the position measured from the head of its list; the minimum depth is $1$. If a list exists at the index of insertion, then an object inserted at that index is added to the ToL.\footnote{A design where objects are inserted at the HoL are also vulnerable to ACAs and would result in essentially the same analysis.}

%

In addition to insertion, the hash table also handles query and delete requests. A query (deletion) from a client is said to be a \defn{good query} (\defn{good deletion}); otherwise, it is a \defn{bad query} (\defn{bad deletion}). Good queries (deletions) cannot be distinguished from bad queries (deletions).  Clients only issue queries for good objects, while the adversary may issue queries for any object; this captures a pessimistic setting where the objects inserted by the adversary are not useful and only serve to degrade the performance of the algorithm.  

Clients may delete good objects, and the adversary may delete bad objects but not good objects. We note that, in practice, the situation may be even more restrictive: a client might only be allowed to delete those objects it inserts. This can be accomplished by having the client share a secret with the server in order to verify ownership of the object. 
\medskip

\noindent{\bf Resource Burning.} Upon receiving a request to insert or query an object the server may issue a \defn{resource-burning (RB) challenge} to the requester (i.e., a client or the adversary). The requester must return a solution to the RB challenge before the corresponding request is satisfied. To specify the RB cost $x$, for any positive integer $x$,  we will refer to an \defn{{\boldmath{$x$}}-hard} RB challenge.

The mechanisms for issuing and verifying RB challenges can be protected from  attack themselves, given their narrow functionality~\cite{waters:new}. Furthermore, significant work has gone into addressing the many practical details of designing and deploying RB challenges, such as handling device heterogeneity, pre-computation attacks, and the reuse of old solutions (see~\cite{ali:foundations,li:sybilcontrol,walfish2010ddos}).
\medskip


\noindent{\textbf{Performance Metrics.}} We use two metrics for gauging performance: (1) the \defn{RB cost} of solving RB challenges and (2) the \defn{latency} for servicing requests.  Regarding (1), the \defn{algorithm's RB cost} is the sum of the hardness values for all RB challenges solved by the clients; likewise, the \defn{adversary's RB cost} is the sum of the hardness values for RB challenges it solves.

To quantify (2), if the request is an insertion, then the latency equals $1$, since we assume that each list maintains a pointer to the ToL. If the request is a query and the object exists in that list, then the latency equals the object's depth in that list; otherwise,  the latency equals the list length at the index where the object would have been stored if it existed in the table.  
\medskip


\noindent{\bf Adversary.} We consider a \defn{Byzantine adversary} that is not constrained to obey protocol and is not computationally bounded.  The adversary has full knowledge of the hash table's configuration, as well as the state of the clients and server (see Section~\ref{sec:query-cost-latency}).  The adversary can instantaneously create as many (bad) objects as it likes that hash to any targeted index of the hash table. In other words, the number of bad objects and the indices into which they are inserted is chosen by the adversary. 

In contrast, the adversary has no control over where good objects are inserted; each good object is inserted into an index chosen uniformly at random from the set of all indices. We consider both the case where (1) the good queries are generated by the adversary (Section~\ref{sec:query-cost-latency}); and (2) the good queries are distributed \uar over all indices (Section~\ref{sec:expected-query-cost}).


\subsection{Main Results}\label{sec:main-results}

For requests, we let {\boldmath{$\numIns$}} be the number of good insertions; let {\boldmath{$\numQue$}} and {\boldmath{$\numDel$}} be the number of good queries and good deletions for objects that exist in the hash table. 

In discussing the hash table, we denote the number of indices in the hash table by {\boldmath{$\htablesize$}}; the maximum number of good objects in any index by {\boldmath{$\ell_M$}}. The number of good objects in index $i$ is {\boldmath{$\ell_i$}}, and  the average number of good objects, $(1/t)\sum_{i=1}^t \ell_i$, is denoted by {\boldmath{$\listave$}}.

Throughput, we use {\boldmath{$\AdvTotal$}} to denote the total RB cost incurred by the adversary. We state our main result below regarding our defense algorithm, \AlgoName.

\begin{myTheorem}\label{thm:main-theorem}
\AlgoName guarantees the following properties:\vspace{-0pt}
\begin{enumerate}[leftmargin=17pt]

\item (Single Requests) Any single insertion has RB cost $O(\sqrt{\AdvTotal} + \listmax)$ and latency $O(1)$.
Any single query or deletion has an RB cost and latency that are each 
$O(\sqrt{\AdvTotal} + \listmax)$.\label{property:single-operation}\vspace{2pt}

\item (Amortized Requests) When $\numIns$, $\numQue$, and $\numDel$ are set by an Byzantine adversary, the total RB cost and the total latency are each $O((\numIns + \numQue + \numDel + \sqrt{(\numIns + \numQue + \numDel)\AdvTotal} \hspace{1pt})\listmax^2).$\label{property:worst-case-costs} \vspace{2pt}

\item (Randomly Queried Indices) Consider $Q$ queries where the corresponding objects belong to indices of the hash table chosen \uar.  For $Q \geq \listmax^2 \AdvTotal$, the average cost per query is $O(\listave)$ in expectation.\label{property:expected-costs}\vspace{1pt}

\end{enumerate}
\end{myTheorem}

\noindent{\bf Discussion.} As mentioned earlier, our results are parameterized by $\listmax$.  For $\numIns$ insertions into a table of size $t$,  $\listmax = O( \lceil \numIns/t \rceil \log t )$
with high probability in $t$ (\defn{w.h.p.}).\footnote{With probability at least $1-t^{-d}$ for some constant $d\geq 1$.}  Notably, for  
$\numIns = O(t)$, it is well known that $\listmax = O(\log t/\log\log t)$ \cite{max-load,kesselheim:load,raab1998balls}.  This case is pertinent, since many applications limit the amount by which their hash table can grow (see ~\cite{Crosby2003AC}), and the number of size increases may be very limited (e.g. see~\cite{czubak2017algorithmic}).\footnote{For example, the table in the Cisco router examined in~\cite{czubak2017algorithmic} has an initial size of $1024$, and can increase to sizes $2048$, $4096$, and $8192$.} From a theory perspective, such limited growth increases the table size by a constant factor, and using the the largest size aligns with our model.

RB is a well-established tool for securing distributed systems~\cite{gupta2020resource}.  We note that the choice of resource to burn likely depends on the specifics of the application. Given this, our algorithm is deliberately agnostic about the  resource burned, such as  computational power~\cite{wang:defending}, bandwidth~\cite{walfish:ddos},  computer memory~\cite{abadi2005moderately,dwork2003memory,dziembowski2015proofs}, and human effort~\cite{von2003captcha,oikonomou2009modeling}.



To provide context for Theorem~\ref{thm:main-theorem}, it is helpful to compare \AlgoName to a standard hash table with chaining.  In the latter, the adversary may create a list that has size linear in the number of bad objects for ``free''. This attack leads to poor latency if good objects reside at the ToL. By comparison, Property~\ref{property:single-operation} bounds improves (roughly) quadratically by bounding the longest list length to be $O(\sqrt{\AdvTotal} + \listmax)$; clearly, this holds for any query, even for objects that do not exist in the table. Another implication of Property~\ref{property:single-operation}  is that when there is little-to-no attack (i.e., when $\AdvTotal \approx 0$), the RB cost and latency are each roughly $O(\listmax)$, which should be small (i.e., logarithmic in $t$) for an appropriately sized table, as discussed above.

Regarding Property~\ref{property:worst-case-costs},  for multiple requests scheduled by a Byzantine adversary, \AlgoName retains an asymptotic advantage when under significant attack. Conversely, when the attack is not large relative to  $\numIns+\numQue+\numDel$,  \AlgoName has RB cost and latency proportional to this number of requests  and $\listmax^2$. In contrast, in a standard hash table, the adversary can amplify its attack by forcing multiple requests involving a linear-sized chain.

Finally, Property~\ref{property:expected-costs} provides bounds on the expected performance under a sequence of queries that map to indices selected \uar. Specifically, the expected cost for $Q$ such good queries is $O(Q\listave +   \ell_M \sqrt{Q\AdvTotal})$. Thus,  if that expectation holds, then the average cost per query is $O(\listave)$ when $Q$ is large relative to $\AdvTotal$ and $\listmax$. When  $\listave = O(1)$, this implies that the average cost per query is $O(1)$ in expectation. Interestingly, this is comparable to the expected $O(1)$ latency per query in a standard hash table.

\subsection{Technical Overview}\label{sec:technical-overview}
At a high level, our analysis relies on upper bounding the algorithm's cost and lower bounding the adversary's cost.  Below we sketch how to do this first for insertions, and then for queries and deletions.

\medskip
\noindent
\textbf{Insertion Costs.} We define a targeted index to be an index where there is at least one bad object and at least one good object.  Then we lower bound the adversary's cost as a function of the number of bad objects inserted into targeted indices (Lemma~\ref{l:max-alg-attack-indices}).

Next, we upper bound the total cost of good insertions as a function of the number of objects in targeted indices, noting that this cost is maximized when the bad objects are distributed as uniformly as possible across such indices (Lemma~\ref{l:attack-indices}).  We pessimistically assume that good insertions come after bad insertions, and that there are $\listmax$ good insertions in every targeted index.\medskip


\noindent{\bf Why Use Move-to-Front?}  An analysis of the longest list (Lemma~\ref{l:chain-length}) shows that the  worst-case latency per query is $O(\sqrt{\AdvTotal}+ \listmax)$. While this significantly improves over the linear latency---for example, where the adversary places all objects in a single list---that can arise in undefended hash tables, there is still room for improvement. To see why, consider that even if the adversary ceases its attack, good objects will remain near the tails of their respective lists, leading to persistently poor query latency. By moving queried objects to the head of their respective lists, we can improve their latency in subsequent queries.  

The classic move-to-front (MTF) heuristic~\cite{bentley1985amortized,hester1985self,rivest1976self} is known to improve performance in chained hash tables when they are not under attack~\cite{zobel2001memory,Nikolas2011Redesigning,Song2017RwHash}. Our motivation for using MTF in our adversarial setting is that a substantial improvement may be attained over multiple queries, since good objects can ``skip the line'' in long lists that contain mostly bad objects.


However, the adversary can cause trouble for MTF in the following manner. When the adversary queries a bad object, it is moved to the front of the list. This increments the depth of a number of good objects as large as the depth of the bad object prior to being moved to the front; this increases the query latency of these good objects. We can discourage this bad behavior by charging for a query, but how much should we charge?  Intuitively, a reasonable charge would be the depth of the queried object.\medskip


\noindent{\bf Analysis of Charging by Depth for Queries.}
To see why this is the correct charging scheme, consider a list composed of bad objects, except for a single good object $o$ at the HoL. In order to increase the depth of $o$ by $d$, the adversary must pay for $d$ bad queries. Observe that each bad query must be for a bad object with larger depth than $o$; otherwise, querying the bad object does not increase $o$'s depth. Under our charging scheme, the adversary  pays at least $\sum_{j=1}^d (j+1) = \Theta(d^2)$. Then when the algorithm next queries  $o$, it will pay an RB cost of $\Theta(d)$, and the query requires $\Theta(d)$ latency. Thus, the algorithm obtains a quadratic advantage, similar to what is achieved for our bound on insertion costs.


This charging scheme motivates the name \AlgoName\footnote{At the risk of ruining a pun via explanation, a depth charge also is a defense against subs/subterfuge.} and it guarantees that the adversary must spend continually in order to keep good objects at large depth in the list.

\medskip
\noindent
\textbf{The Amortized Analysis.} A major technical challenge of our paper is to formalize the above intuition in the general case---with multiple lists, each with potentially multiple good objects.  This analysis is challenging, since both bad and good queries can increase the depth of multiple good objects in a list. Over all lists, we need to track the depth of all good objects over a sequence of requests. We highlight that must account not only for queries---although they are what increases the depth of an object---but also insertions and deletions. Fortunately, insertions do not increase depth of other objects, given that objects are added to the ToL, so our bound on insertion cost (discussed above) can be used. As for deletions,  we can treat them as queries, since they are no worse in terms of increasing depth.

One main analytic tool used is amortized analysis; in particular, the accounting method~\cite{CLRS}. Each good object is given a (conceptual) wallet into which the algorithm makes deposits for each request that increases the depth of that object. The payments ensure a key invariant: the depth of a good object is never more than the number of dollars in its wallet. Therefore, an object's wallet always contains enough dollars to cover the cost of its next query. 
Over a sequence of requests, the total number of dollars deposited into all wallets is an upper bound on both the algorithm's RB cost and latency.


How can we relate the number of dollars deposited into wallets to the adversary's cost? This is addressed formally in  Lemma~\ref{lem:subset-queries}; however, to gain insight, let us extend our example to $q_i\geq 1$  good queries in a single list at index $i$.  Prior to each good query, there are $d_r$ bad queries that increase the depth of at most $\listmax$ good objects by $d_r$, for $r=1, ..., q_i$.  The resulting number of dollars that the algorithm places into the wallets of the corresponding $\listmax$ good objects is $\mathcal{A}_i \leq \listmax\sum_{r=1}^{q_i} d_r$, while the adversary's cost is $\AdvTotal_i \geq \sum_{r=1}^{q_i} d_r^2 = \Omega( (1/q_i) (\sum_{r=1}^{q_i} d_r)^2)$ by
 Jensen's inequality for concave functions. Therefore, the number of dollars deposited into wallets for objects in the list at index  $i$ is $\mathcal{A}_i  = O(\listmax\sqrt{q_i \AdvTotal_i})$. To this, we add the algorithm's cost for insertions, denoted by $\InsertCost_i$, to get an upper bound on all requests involving this list.

Finally, in Lemma~\ref{lem:query-cost-total} and Corollary~\ref{cor:total-worst}, we sum up the costs to the algorithm over all lists. Our previous bound on the insertion costs handles the sum of the $\InsertCost_i$ terms. To simplify    $O(\sum_i \listmax\sqrt{q_i \AdvTotal_i})$, we apply the Cauchy-Schwarz inequality to get an upper bound of $O(\listmax \sqrt{\numQue\AdvTotal})$, where $\sum_i q_i = \numQue$ is the total number of queries and $ \sum \AdvTotal_i \leq \AdvTotal$, where $\AdvTotal$ is {\it total} adversarial cost. Together, these bounds yield the expression in Property~\ref{property:worst-case-costs}.\medskip

\noindent{\bf Randomly Queried Indices.} Our result for randomly queried indices does not follow directly from Property~\ref{property:worst-case-costs}. Instead, our argument (Lemma~\ref{lem:expected-queries}) leverages the bound for a single list (Lemma~\ref{lem:subset-queries}) in order to express the total cost from the randomly queried indices as a function of $E[Q_i]$ and $E[\sqrt{Q_i}]$. The latter is the more complicated term, which is handled  by the application of Jensen's inequality for the expectation of concave functions, which shows that $E[\sqrt{Q_i}] \leq \sqrt{E[Q_i]}$.  Using the fact that $E[Q_i] = Q/t$, and summing the terms over all lists, yields the expression in Property~\ref{property:expected-costs}.

\subsection{Related Work}\label{sec:related-work}

In this section, we summarize work on RB-based defenses for a variety of attacks. Next, we discuss results from the literature on ACAs, with a focus on prior results for hash tables.\medskip

\noindent{\bf Defenses using Resource Burning.} The use of RB as a tool for solving security problems spans several decades (e.g. see the surveys~\cite{ali:foundations, gupta2020resource}).  RB-based defenses arise in many contexts, such as spam mitigation~\cite{dwork2003memory,dwork:pricing},  wireless networks~\cite{gilbert:sybilcast}, peer-to-peer systems~\cite{li:sybilcontrol,borisov:computational}, blockchains~\cite{lin2017survey}, the Sybil attack~\cite{Gupta_Saia_Young_2021}, and denial-of-service attacks~\cite{walfish2010ddos,parno2007portcullis,chakraborty2023bankrupting}.

\medskip

\noindent{\bf Prior Defenses for ACAs.} Many other common data structures and algorithms are vulnerable to ACAs, such as linked lists~\cite{atre2022surgeprotector}, quicksort~\cite{khan2005prevention, mcilroy1999killer}, cardinality sketches~\cite{Pedro2020hll}, pattern matchers~\cite{kirrage2013static,Kedar2010regex,berglund2014analyzing}, cuckoo filters~\cite{Pedro2020Cuckoo}, and bloom filters~\cite{Pedro2020Pollution}. As a result, AC attacks can impact common applications: networked applications~\cite{Chang2009coma,atre2022surgeprotector}, firewalls~\cite{czubak2017algorithmic}, web services~\cite{altmeier2016adidos}, PDF compressors~\cite{hauke2019Dos}, TCP reassembler~\cite{Sarang2005TCP,Juha2018CVE}, and intrusion detection systems~\cite{Crosby2003AC}.

In the context of hash tables, the prior literature on defending against ACAs falls into the three general categories discussed below.\medskip


\noindent{{\underline{\it (1) Choice of Hash Functions}}.} Crosby et al.~\cite{Crosby2003AC} showed the first ACA on hash tables, which caused a server to drop over 70\% of queries. The authors proposed two techniques to mitigate ACAs: (a) adding a secret value as a parameter to the hash function, and (b) using universal hash functions (UHFs). The usage of UHFs can minimize the number of collisions, but UHFs can add computational overhead on the server side. Furthermore, Yosef et al.~\cite{Yosef2007RemoteAC} demonstrated an ACA against hash tables despite the use of a secret value; the authors suggest that the secret-key length should be increased (beyond 32 bits) or be changed frequently. In a similar vein,  SipHash~\cite{aumasson2012siphash} uses a secret key (known only to the server), which is used as input to the hash function. Unfortunately, a secret key may be compromised via side-channel attacks~\cite{Olekšák2022siphash} or, in distributed  settings, by an adversary who controls one or more of the servers.  Finally, perfect hashing is a technique that guarantees no collisions (see~\cite{lu2006perfect,cercone1988finding,majewski1996family}). However, constructing perfect hash functions is time consuming and requires knowing the set of objects to be hashed, which is not always available.
\medskip 

\noindent{{\underline{\it (2) Application-Specific Defenses}}.} Many defenses against ACAs are application-specific. For example, PHP limits the number of GET and POST HTTP requests so that the adversary cannot request to store many bad objects in a hash table~\cite{PHP-online}. Another approach is the use of caching to store pre-computed results of expensive hash table lookups~\cite{zhang2019leveraging,Metreveli2012CPHASH,bender2011don}.  \medskip

\noindent{{\underline{\it (3) Switching to Deterministic Data Structures}}.} Another method for defending against ACAs is to adopt deterministic data structures with strong worst-case performance guarantees. For example, a deterministic skip list~\cite{Munro1992detskiplist} performs each insertion and  each query with a worst-case bound that is logarithmic in the number of objects. 
However, this is  inferior to the performance of a hash table, which have constant expected time per query in the absence of attack. If attacks are likely to occur over a minority of the system lifetime, then using a deterministic data structure is costly. 

More generally, deterministic data structures incur theoretical and/or practical costs exceeding that of their randomized equivalents; for example, this shortcoming is acknowledged~\cite{czubak2017algorithmic} in regards to B-trees, AVL-trees, and red-black trees, which offer worst-case logarithmic guarantees. Maintaining both a deterministic and randomized data structure might provide the advantages of both options, but such redundancy is likely to be expensive. In contrast, our algorithm's costs adapts to the degree of attack---in particular, growing slowly in the amount spent by the adversary---which allows for low cost when the attack is absent/small, and giving a favorable relative cost when the attack is large.
\medskip

\noindent{\bf Compatibility with Prior Defenses.}  Our algorithmic results may be used in conjunction with many prior solutions. For example, \AlgoName can be used alongside methods that use stronger hash functions and secret keys, and also within application-specific defenses---these approaches are not mutually exclusive. Given that there is no single approach that can completely protect against ACAs, having multiple complementary tools for defense can be useful.

\section{Our Algorithm}\label{s:algorithmg-desciption}

The pseudocode for our algorithm (with deletions omitted), \AlgoName, is presented in Figure~\ref{alg:ouralg} and is assumed to be executed by the server. Below, are the requests supported.\smallskip

\noindent{\bf Insertions.} Upon receiving an insertion request, the server responds with an RB challenge whose hardness equals $L_i+1$, where $L_i$ is  list length at index $i$. If the server receives a valid solution to this challenge, then the object is inserted at the ToL, which is assumed to require constant latency. A pointer is assumed to be kept to the ToL in order to give $O(1)$ latency per insertion. \smallskip

\noindent{\bf Queries.} Upon receiving a query request for an object, the server calculates the index $i$ where the object should be stored and traverses that list starting from the HoL. If the object is found, then the server responds with a $\Delta$-hard RB challenge, where $\Delta$ is the object's depth. Otherwise, the server  discovers that the object does not exist by traversing the entire list and then issues  an $L_i$-hard RB challenge. In the latter case, imposing a cost mitigates spurious requests by the adversary for non-existent objects, while the cost for such requests from clients can be viewed as the price for a membership test.

If a valid solution is received and the object exists in the table, then the server services the query and also performs a move-to-front operation by repositioning the queried object to the head of its corresponding list. Otherwise, the queried object does not exist in the table, and responds that the object was not found.
\medskip

\noindent{\bf Deletions.} A deletion  is performed almost identically to a query. However, in the case where the object is located, the object is deleted rather than being moved to the HoL. For ease of presentation, we omit deletions from the pseudocode in Figure~\ref{alg:ouralg}.

\begin{figure}[t!]
\centering
\begin{tcolorbox}[standard jigsaw, opacityback=0]
\begin{minipage}[h]{0.99\textwidth}
\noindent{\textsc{\large \textbf{\AlgoName} }} 
\medskip

\noindent Insert at index $i$:
\vspace{2pt}
\begin{itemize}[leftmargin=15pt]

\item Respond with an $(L_i+1)$-hard RB challenge, where $L_i$ is the list length at index $i$.\vspace{2pt} 

\item If the requester solves the RB challenge, then insert the object at the tail of the list at index $i$.
\end{itemize}

\noindent Query object at index $i$: \vspace{2pt}
\begin{itemize}[leftmargin=10pt]
\item Traverse the list at the index $i$.  If the object is found, then issue a $\Delta$-hard RB challenge to the requester, where $\Delta$ is the object's depth; else, respond with an $L_i$-hard challenge. \vspace{2pt}

\item If the requester solves the RB challenge, then if the object exists, service the query and move the object to the HoL; else, respond that the object was not found.

\end{itemize}

\end{minipage}
\end{tcolorbox}
\vspace{-12pt}\caption{Pseudocode for \AlgoName.}
\label{alg:ouralg}
\vspace{-12pt}
\end{figure}



\section{Analysis}

Our analysis of \AlgoName is presented in three pieces. First, in Section~\ref{sec:insertion-cost}, we analyze the RB cost and latency for insertions; this is a stepping stone to proving bounds on sequences of requests.  Second, in Section~\ref{sec:query-cost-latency}, we provide a bound on the longest list length (Lemma~\ref{l:chain-length}), which is used to establish Property~\ref{property:single-operation}.  We then use an amortized analysis for a sequence of queries, which is combined with our bound on insertions to prove Property~\ref{property:worst-case-costs} (in Corollary~\ref{cor:total-worst}). Third, in Section~\ref{sec:expected-query-cost}, we bound the expected RB cost and latency for a sequence of queries that occur in indices selected \uar, which allows us to establish Property~\ref{property:expected-costs} (in Lemma~\ref{lem:expected-queries}).




\subsection{Insertion Cost}\label{sec:insertion-cost}

In this section, we analyze the algorithm's RB cost over all $\numIns$ good insertions.  Define an \defn{targeted index} to be any index that contains at least one bad object and at least one good object. In the current table, let {\boldmath{$s$}} be the number of targeted indices. 

Unless specified otherwise, our analysis in this section pessimistically assumes that all targeted indices contain $\listmax$ good objects; this can only increase the cost to the algorithm. 


\begin{lemma}\label{l:max-alg-attack-indices}
Suppose the adversary places $b$ bad balls in the $s$ targeted indices. Then, the algorithm's RB cost for insertions into the attacks indices is at most $s \listmax^{2} +   b\listmax$.
\end{lemma}
\begin{proof}
Let $x_i$ be the number of bad objects placed by the adversary into the $i$-th targeted index, where $i=1, ..., s$.  Fix any particular index  $i$, the algorithm's cost for this index is at most:
\begin{align*}
\sum_{k=1}^{\listmax}(x_i + k) 
& < \frac{\listmax^{2}}{2} + \listmax x_i.
\end{align*}
Using the above bound, the algorithm's insertion cost over all targeted indices is at most:
\begin{align*}
\sum_{i=1}^{s} \left(\listmax^{2} + \listmax x_{i}\right) & \leq s \listmax^{2} + \listmax \sum_{i=1}^{s}x_i \\
& = s \listmax^{2} +  b\listmax
\end{align*}
\noindent where the second line follows from noting that $\sum_{i=1}^{s}x_i = b$.
\end{proof}

\begin{lemma}\label{l:min-adv-attack-indices}
Suppose that the adversary places $b\geq 1$ objects in $s\geq 1$ targeted indices. Then,  $\AdvTotal \geq \frac{b^2}{8s}$.
\end{lemma}
\begin{proof} 
Assume that the adversary's bad objects are all added before any good objects are added to the table; this only reduces the adversary's cost. Furthermore, observe that the adversary's cost from the placement of bad objects in targeted indices is minimized when these $b$ objects are spread as evenly as possible over the $s$ indices. To see this, we describe two cases.\medskip

\noindent\textbf{Case \boldmath{$b~(mod~\attackindices)= 0$}}. Consider any two indices, each with $x =  b/\attackindices$ bad objects. In this case, the adversary's cost is 
$2 \sum_{i=1}^{x} i$. In contrast, if we move $p$ bad object, where $p\in[1,x]$ from one of these indices to the other, the adversary's cost is $\sum_{j=1}^{x-p} j + \sum_{k=1}^{x+p} k$. Note that the first cost minus the second cost is:
\begin{align*}
&2 \sum_{i=1}^{x} i - \left(\sum_{j=1}^{x-p} j + \sum_{k=1}^{x+p} k\right)\\
& = \left( (x-p+1) + ... + x\right) - \left( (x+1) + ... + (x+p) \right) \\
&< 0.
\end{align*}

\noindent Therefore, deviating from the case where all indices have the same number of bad objects will increase the adversary's cost.\smallskip

\noindent\textbf{Case \boldmath{$b~(mod~\attackindices)\neq 0$}}. In this case, there will be indices with $x$ objects and at least one index with $x+1$ bad objects; thus, there can be at most a difference of $1$ bad object between any two indices.  Consider any a ``small'' index and a ``large'' index with $x$ and $x+1$ bad objects, respectively. In this case, the adversary's cost is 
$\sum_{h=1}^{x} h + \sum_{i=1}^{x+1} i$. 

Moving $p$ bad objects from the small index to the large index means that the adversary's cost is  now $\sum_{j=1}^{x-p} j + \sum_{k=1}^{(x+1)+p} k$. Note that the first cost minus the second cost is:
\begin{align*}
&\sum_{h=1}^{x} h + \sum_{i=1}^{x+1} i - \left(\sum_{j=1}^{(x+1)+p} \hspace{-7pt} j + \sum_{k=1}^{x-p} k\right) \\
& = \left( (x-p+1) + ... + (x+1)\right) - \left( (x+1) + ... + (x+1+p)\right)\\
&< 0.
\end{align*}

\noindent  Again, deviating from the case where all indices have the same number of bad objects will increase the adversary's cost.  

Given this case analysis, the adversary's cost over the $s$ targeted indices is at least:
\begin{align*}
s \sum_{i=1}^{\lfloor b/s \rfloor} i & \geq s \int_{0}^{\lfloor b/s \rfloor} i~di 
\end{align*}
\begin{align*}
&= \left(\frac{s}{2}\right)\left(\lfloor b/s \rfloor\right)^{2}\\
& \geq \left(\frac{s}{2}\right)\left( \max\{1, (b/s)-1\}\right)^{2}\\
& \geq \left(\frac{s}{2}\right)\left( \frac{b}{2s} \right)^{2}\\
& = \frac{b^2}{8s}
\end{align*}
\noindent where the first line follows since $i$ is a monotonically increasing function, and the third line holds since $b\geq s$ by definition of targeted indices and $\lfloor x \rfloor \geq x-1$.  The fourth line follows by noting that, if $\max\{1, (b/s)-1\}=1$, then $b/s \leq 2$ and so $b/(2s) \leq 1$, which justifies the inequality. Else, if $b/s - 1 > 1$, which implies $b/(2s) > 1$ iff $(b/s)-(b/2s) > 1$ iff  $(b/s)-1 > b/(2s)$, which again justifies the inequality (although, it is strict in this case). 
\end{proof}

\begin{lemma}\label{l:attack-indices}
The RB cost to the algorithm for insertions into targeted indices is
$O\left( \listmax^2\sqrt{\attackindices \mathcal{B}} \right)$.
\end{lemma}
\begin{proof}
By Lemma~\ref{l:min-adv-attack-indices}, $\AdvTotal\geq  \frac{b^2}{8s}$ for placing $b$ objects into targeted indices. Solving for $b$ yields $b\leq \sqrt{8 s \AdvTotal}$.  Next, we use Lemma~\ref{l:max-alg-attack-indices}, which shows that the RB cost to the algorithm due to the targeted indices is at most $s \listmax^{2} +   \listmax b$. Thus, the algorithm’s cost for good objects in targeted indices is at most:
\begin{align*}
      &= s \listmax^{2} +  b\listmax\\ 
      & \leq  \attackindices\listmax^{2} + \listmax \sqrt{8 s \AdvTotal}  \\
     & = O\left( \listmax \sqrt{s \AdvTotal} + s\listmax^{2}
 \right)
\end{align*}
\noindent where the second step holds by substituting the upper bound on $b$.  Noting that $\AdvTotal \geq s$ yields the claim.
\end{proof}

Define a \defn{good index} to be an index containing only good objects. Having analyzed the cost to targeted indices, we now analyze the additional cost to the algorithm due to good indices.

\begin{lemma}\label{l:good-indices}
With high probability, the RB cost to the algorithm for good insertions into good indices is  $O( \numIns \listmax^2) $.
\end{lemma}
\begin{proof}
There are at most $\numIns$ good indices, each with $\listmax$ good objects. The resource burning cost to the algorithm for at most $\htablesize$ such indices is at most:
    \begin{align*}
        \numIns\left( \sum_{i=1}^{\listmax}\hspace{-0pt}i \hspace{7pt}\right) & =  O\hspace{0pt}\left( \numIns 
 \listmax^2 \right)
    \end{align*}
    \noindent which completes the argument.
\end{proof}

\noindent{}We can now bound the total RB cost to the algorithm over the $\numIns$ insertions.

\begin{corollary}\label{cor:cost-single-table}
The total RB cost to the algorithm for good insertions is:
$$O\hspace{-3pt}\left(\left( \sqrt{\numIns  \AdvTotal}  + \numIns\right)\listmax^2\right).$$ 
\end{corollary}
\begin{proof}
This follows directly by adding up the algorithm's cost incurred by all targeted indices and good indices as derived in Lemmas~\ref{l:attack-indices} and~\ref{l:good-indices}, respectively, and noting that $\numIns \geq s$.
\end{proof}

\subsection{Single and Amortized Requests}\label{sec:query-cost-latency}

We start by obtaining an upper bound on the longest list that can be created by the adversary, which in turn, provides an upper bound on the RB cost and latency for any single query. Note that this bound holds regardless of whether the corresponding object exists in the hash table, which establishes Property~\ref{property:single-operation} in Theorem~\ref{thm:main-theorem}.

\begin{lemma}\label{l:chain-length}
The maximum number of bad objects in any list is $O(\sqrt{\AdvTotal})$ and, with high probability, the RB cost and latency for any single query is 
$O\left( \sqrt{\AdvTotal}+ \listmax\right)$.
\end{lemma}
\begin{proof}
The cost to the adversary is minimized if its bad objects are inserted in a list ahead of any good objects; thus, the cost for $b$ bad objects is at least $\sum_{j=1}^{b} j$. Given that the adversary spends $\AdvTotal$, the maximum number of bad objects $b$ that can be placed in an index  satisfies the following equation:
\begin{align*}
\AdvTotal  \geq  \sum_{j=1}^{b} j\\
> b^2/2
\end{align*} 
and solving for $b$ yields:
\begin{align*}
b < \sqrt{2\AdvTotal}. 
\end{align*}
\noindent Noting that there are at most $\listmax$ good objects in this list establishes the maximum number of bad objects in the list. Finally, since the RB cost and latency for a query are both equal to the depth of the associated object, the claim follows.
\end{proof}

Next, we examine the cost to our algorithm under a sequence of {\boldmath{$\numQue$}} good queries, whose corresponding objects exist in the table.  We analyze the MTF heuristic to show that  the adversary must continually incur an RB cost in order cause bad latency for $\numQue$.\medskip





\noindent{\bf Setup and Argument Overview.}  We first focus on a single list and the subsequence of $q$ good queries that involve this list: we denote these queries by  $Q_1, Q_2, \dots, Q_q$ for the queried (good) objects $o_1, o_2,  \dots, o_{q}$. We can later aggregate the costs to the algorithm over all lists to arrive at our final claim.

A complication arises due to the changing position of good objects over time.  For example, once a good object $o_r$ is queried under $Q_r$, for $1 \leq r\leq q$, we must keep track of $o_r$, so that we can charge the algorithm the correct amount if it is queried again later; simply assuming the object has an RB cost and latency equal to the list length would result in poor bounds. Additionally, queries prior to $Q_r$ do not only increase the depth of $o_r$, but also every other object in the list (except for the object at the ToL), thus increasing their query cost and latency. This increase in depth is illustrated in Figure~\ref{fig:query-analysis}. 

As discussed in Section~\ref{sec:technical-overview}, we use amortized analysis to handle such complications. Specifically, we use the accounting method, where we track the algorithm's cost by letting each good object have a conceptual ``wallet''; in this section, we speak of cost in terms of generic dollars.  When queried, the good object pays for this query with dollars from its wallet. The total amount of money placed in the wallets of all good objects provides an upper bound on the algorithm's RB cost and  latency for queries.

Initially, each wallet holds a dollar amount equal to the {\it insertion} cost in its list. For the purposes of our accounting-method analysis, this means that when $o_r$ is originally inserted, the algorithm conceptually pays $L+1$ dollars for the insertion and another $L+1$ dollars as a down-payment towards its next query, where $L$ is the list length immediately prior to the insertion of $o_r$. Thus, the RB cost for the first query of $o_r$ is at least partially paid for, since $o_r$'s wallet holds dollars equal to its depth when inserted.  These extra $L+1$ dollars are charged to the insertion of $o_r$; this is captured by Corollary~\ref{cor:cost-single-table}.
\medskip

\noindent{\bf Defining Rounds.} To analyze attacks on the use of MTF in a single list, we consider a sequence of \defn{rounds} also indexed by $r$, for $r=1, ..., q$.  Round $r$ starts with the adversary selecting an integer value $d_r\geq 0$ and moving $d_r$ bad objects to the HoL via $d_r$ queries, which the adversary pays for.\footnote{The adversary can never increase the depth of a good object to more than the its corresponding list length; however, the adversary can perform as many bad queries as it wishes, i.e. it can set  $d_r$ to any non-negative value.}  For each good object in this list, the algorithm places a dollar amount into each wallet equal to the increase in the depth of the corresponding good object, which is upper bounded by $d_r$. Then, query $Q_r$ is executed, which brings the queried good object $o_r$ to the HoL and reduces $o_r$'s wallet to zero. Next, we insert an additional $1$ dollar into each of the wallets of all good objects in the list whose depth increased by $1$ by bringing $o_r$ to the HoL, and also place an additional $1$ dollar into $o_r$'s wallet. After these actions are completed, round $r$ ends. 

Payments by the algorithm at the end of each round allow us to maintain the following invariant in our amortized analysis: {\it At the end of each round, for every good object, the amount of money in the good object's wallet is at least equal to its depth}. We leverage this invariant in our analysis of the algorithm's RB cost and latency. 

Finally, over all rounds, the adversary may schedule good and bad insertions arbitrarily. These insertions do {\it not} increase the depth of any good object in a list, since objects are inserted at the ToL. Consequently, the algorithm's RB cost and latency for good insertions in any list can be accounted for separately in our analysis. 
\medskip

\begin{figure*}[t]
  \centering
  \includegraphics[width=0.8\textwidth]{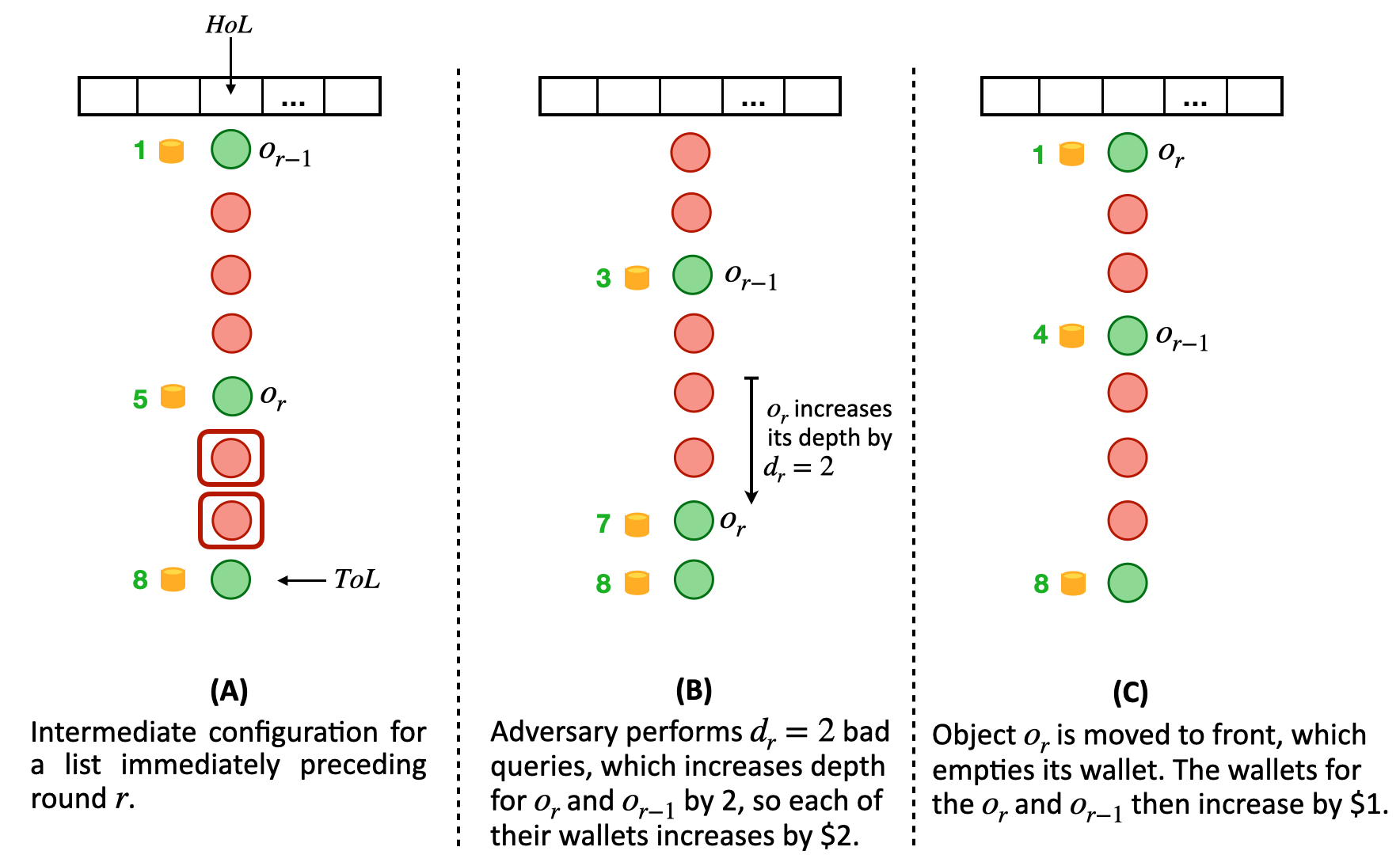}
  \caption{An illustration of the query analysis for some intermediate round $r$. Green and red balls represent good and bad objects, respectively. The amount of money in a wallet is depicted by the number of coins. {(A)} This is the hash table's state at the end of previous round  $r-1$, where object $o_{r-1}$ resides at the HoL and the good objects hold \textdollar 1, \textdollar 5, and \textdollar 8  in their respective wallets. {(B)}  The adversary chooses $d_r=2$ and so executes $2$ bad queries, which increases the depth of the first two good objects by $2$ (and does not impact the good object already at the ToL). {(C)} Query $Q_r$ for object $o_r$ is executed, which empties its wallet corresponding to the cost of $7$ for this query. This results in a depth of $1$ for $o_r$, while increasing the depth of the second good object from the HoL by $1$; therefore, the algorithm  adds \textdollar 1 to each of their wallets (but not to the wallet of the good object at the ToL). Thus, round $r$ ends with each good object holding an amount of money at least equal to its current depth.} 
  \label{fig:query-analysis}
\end{figure*}

Our next lemma considers the subsequence of queries in a single list of the hash table. We note that deletions are no worse than any queries, since deletions can only decrease the depth of a good object. Thus, for ease of presentation, our analysis only argues about insertions and queries, even though the statement of our final result will include deletions.

\begin{lemma}\label{lem:subset-queries}
Consider any fixed list at index $i$ in the hash table and suppose this list is involved in $q_i$ good queries whose corresponding objects exist in the table. Let $\InsertCost_i$ be the algorithm's total RB cost to insert the good objects in list $i$. Let $\AdvTotal_i$ be the cost to the adversary for bad queries and bad insertions in list $i$. For all of the $q_i$ queries, the total RB cost and total latency for the algorithm is at most:
     $$\InsertCost_i + \ell_i \left(q_i + \sqrt{2 q_i \AdvTotal_i}\right).$$
\end{lemma}
\begin{proof}
Our aim is to guarantee that, prior to round $r\geq 1$, each good object has a number of dollars in its wallet equal to its depth. Given this, we then argue that the number of dollars in each object's wallet can pay for the cost of querying the object; notably, this cost can be either  RB cost or latency, since they are both equal to the object's depth.
\medskip

\noindent{\bf Round 1.} We first prove that, for each good object, the depth is at most the number of dollars in the corresponding object's wallet. Initially, sometime prior to $Q_1$, $o_1$ is be inserted (since, by assumption, it exists in the table when $Q_1$ is executed). The wallet of $o_1$
contains a number of dollars equal to its depth when $o_i$ is inserted. This is done by having $o_i$ pay  $2(L+1)$ when inserted, where $L$ is the list length immediately prior to the $o_1$'s insertion. The first $L+1$ dollars pay for the insertion, while the \defn{extra} $L+1$ dollars are held in the $o_1$'s wallet to help pay for the cost when it is next queried. 

Any increase in depth experienced by good objects due to $d_1$ bad queries in round $1$ results in a matching number of dollars added to each wallet. Thus, when $Q_1$ is executed, $o_1$'s wallet has sufficient funds to pay for the latency of the query. Object $o_1$ moves to the HoL, and every good object whose depth increased by $1$, along with $o_1$, has $1$ dollar added to its corresponding wallet. These deposits to the wallets ensures that the invariant holds at the end of round $1$.\medskip

\noindent{\bf Round {\boldmath{$\geq 2$}}.} At the end of round $r-1$, for $r\geq 2$, each good object in the list holds a number of dollars at least equal to its depth. Thus, in round  $r$, $o_r$ has sufficient funds in its wallet to pay for the RB cost of $Q_r$. The adversary's $d_r$ bad queries increase the depth of each good object by at most an additional $d_r$, and $Q_r$ results in all other good objects increasing their depth by at most $1$. Since the algorithm puts dollars in each good objects' wallets equal to the corresponding increase in depth due to bad queries, the invariant holds at the end of round $r$. 

\medskip

\noindent{\bf Total Cost.} 
The algorithm pays the following.  First, the cost for all good insertions is $\InsertCost_i$. Second, the algorithm pays for all the increases in depth over all rounds for all $\ell_i$ good objects in this list, which amounts to at most $\ell_i\sum_{r=1}^{q_i} (d_r+1)$ dollars.


In contrast, the total RB cost to the adversary is at least:
\begin{align}
 \AdvTotal &\geq \sum_{r=1}^{q_i} \sum_{j=1}^{d_r} j \nonumber \\
 & \geq \frac{1}{2}\sum_{r=1}^{q_i} d_r^2 \nonumber\\
& \geq \frac{1}{2q_i}\left(\sum_{r=1}^{q_i} d_r\right)^2\label{eqn:adv-query-cost}
\end{align}
\noindent where the last line follows from Jensen's inequality for convex functions. By substituting  into the algorithm's cost, we have that the algorithm pays at most:
\begin{align*}
\InsertCost_i + \ell_i\sum_{r=1}^{q_i} (d_r+1) &=  \InsertCost_i + \ell_i q_i + \ell_i\sum_{r=1}^{q_i} d_r\\
& =  \InsertCost_i + \ell_i q_i +\ell_i \sqrt{2 q_i \AdvTotal_i}
\end{align*}
\noindent where the second line follows by solving for $\sum_{r=1}^{q_i} d_r \leq \sqrt{2q_i\AdvTotal_i}$ in Equation~\ref{eqn:adv-query-cost}.  Since $\InsertCost_i$ is measured in RB cost, this concludes the bound on RB cost.

To derive total latency, recall that the RB cost for an insertion equals the depth of the object being inserted. In other words,  $\InsertCost_i$ equals the sum of the depths of the good objects when they are inserted.  Thus, the extra dollars can also be viewed as being stored in $o_i$'s wallet to help pay the latency when the object is next queried. This leads to the same bound on latency.
\end{proof}

\noindent We can now account for the algorithm's total RB cost and total latency for all good queries $\numQue$.

\begin{lemma}\label{lem:query-cost-total}
The total RB cost and the total latency of the algorithm due to the $\numQue$ queries is:
$$O\left( \left(\numIns + \numQue + \sqrt{(\numIns + \numQue)\AdvTotal} \right) \listmax^2\right).$$
\end{lemma}
\begin{proof}
Let $S$ denote the indices of the hash table where at least one good query takes place. For $i\in S$, $q_i$ is the number of good queries that occur in this list; $\InsertCost_i$ is algorithm's RB cost to insert the good objects in list $i$; and $B_{i}$ be the amount spent by the adversary on bad queries in this list.

By Lemma~\ref{lem:subset-queries}, over all good queries, the total RB cost and the total latency are each at most:
\begin{align*}
&\hspace{9pt}\sum_{i\in S} \left( \InsertCost_i + \ell_i \left(q_i + \sqrt{2q_i\AdvTotal_i}\right)  \right) \\
& \leq \sum_{i\in S} \left( \InsertCost_i + \listmax \left(q_i + \sqrt{2q_i\AdvTotal_i}\right)  \right) 
\end{align*}
\begin{align*}
&\leq  \left( \sum_{i\in S}\InsertCost_i\right)  + \left(\listmax \sum_{i\in S} q_i \right) + 
 \listmax\sqrt{ 2\left( \sum_{i\in S}B_i  \right) \left(  \sum_{i\in S}q_i  \right)   }\\
 &\leq  \left(  \sum_{i\in S}\InsertCost_i\right)  +   \listmax \left( \numQue + \sqrt{ 2 \AdvTotal \numQue }\right)\\
&= O\left(\left(\numIns +\sqrt{\numIns\AdvTotal}\right)\listmax^2 +
\left(\numQue + \sqrt{\numQue\AdvTotal  }\right)\listmax\right)
\end{align*}
where the first line follows since $\ell_i \leq \listmax$. The second line is obtained via the Cauchy–Schwarz inequality $\left(\sum_{i\in S} \sqrt{B_i} \sqrt{q_i}\leq \sqrt{\sum_i B_i \sum_i q_i }~\right)$. The third line is derived from noting $\sum_{i\in S}q_i$ $= \numQue$ and $\sum_{i\in S}B_i \leq \AdvTotal$. The fourth line follows from Corollary~\ref{cor:cost-single-table}, which states that $\sum_{i\in S}  \InsertCost_i = O(\numIns +\sqrt{\numIns\AdvTotal})$. 
\end{proof}

\noindent We are now ready to bound the total RB cost and the total latency for all good insertions, along with all good queries and good deletions whose corresponding objects are in the hash table. Corollary~\ref{cor:total-worst} establishes the expression in Property~\ref{property:worst-case-costs} of Theorem~\ref{thm:main-theorem}.

\begin{corollary}\label{cor:total-worst}
Each of the total RB cost and the total latency for the algorithm is:
$$O\left( \left(\numIns + \numQue + \numDel + \sqrt{(\numIns + \numQue + \numDel)\AdvTotal} \right) \listmax^2\right).$$
\end{corollary}
\begin{proof}
Adding the cost from all good insertions in the table, given by Corollary~\ref{cor:cost-single-table}, alters the asymptotic cost given in Lemma~\ref{lem:query-cost-total}. Then, since (as discussed earlier) deletions are no more costly than queries, we may replace $\numQue$ with $\numQue+\numDel$ to obtain the result.
\end{proof}

\subsection{Randomly Queried Indices}\label{sec:expected-query-cost}

Note that Corollary~\ref{cor:total-worst} addresses a challenging setting: deriving worst case bounds where a significant attack may be underway and the good requests are scheduled by the adversary. We conclude this section on a more optimistic note in regards to $Q$ queries that occur in randomly chosen indices. Recall (from Section~\ref{sec:main-results}) that $\ell_i$ is the maximum number of good objects that are ever in bin $i$ and that $\listave = (1/t)\sum_{i=1}^t \ell_i$.  We show that when $Q$ is large relative to $\AdvTotal$ and $\listmax$, the average query cost is $O(\listave)$ in expectation.

This result implies that, if $\listave = O(1)$ and the adversary does not launch a significant attack, then we should expect per-query performance that matches that of standard hash tables in benign settings. The following result establishes Property~\ref{property:expected-costs} of Theorem~\ref{thm:main-theorem}.

\begin{lemma}\label{lem:expected-queries}
Consider $Q$ queries where the corresponding objects belong to indices of the hash table chosen \uar.  For $Q \geq \listmax^2 \AdvTotal$, the average cost per query is $O(\listave)$ in expectation.
\end{lemma}
\begin{proof}
By Lemma~\ref{lem:subset-queries}, the RB cost and latency  for the $i$-th index are each at most:
\begin{align*}
    \ell_i \left(Q_i + \sqrt{2 Q_i \mathcal{\AdvTotal}_i}\right)
\end{align*}
\noindent where $\ell_i$, $Q_i$, and $\AdvTotal_i$ are the number of good objects, number of good queries, and adversarial cost in the $i$-th index. Given that each query occurs in an index selected uniformly at random, in expectation over $Q_i$ the RB cost and latency are each at most:
\begin{align*}
    \ell_i E[Q_i] + \ell_i \sqrt{2\mathcal{\AdvTotal}_i}E[\sqrt{Q_i}] & \leq \ell_i (Q/t) + \ell_i \sqrt{2\mathcal{\AdvTotal}_i} \sqrt{Q/t}
\end{align*}
\noindent where the second step holds since $E[Q_i] = Q/t$, and  by applying Jensen's inequality for concave functions (i.e., $ E[\varphi(X)] \leq \varphi(E[X])$, where $\varphi$ is a concave function and $X$ is a random variable). Summing the above over all bins, we can bound the total RB cost and latency for queries to be at most:
\begin{align*}
\sum_{i=1}^t \ell_i \left (Q/t + \sqrt{2\mathcal{\AdvTotal}_i} \sqrt{Q/t} \right) & =  O( 
 Q \,\listave)  + O \left( \sqrt{2Q/t} \sum_{i=1}^t \ell_i \sqrt{\AdvTotal_i} \right)\\
 & \leq O(Q\,\listave) + O \left( \left(\sqrt{2Q/t} \right)  t \ell_M  \sqrt{\AdvTotal/t} \right) \\
  & \leq O\left(Q\,\listave +   \ell_M \sqrt{Q\AdvTotal}  \right)
\end{align*}
\noindent where the first line follows since $\listave= (1/t)\sum_{i=1}^{t} \ell_i$, the second line follows from again noting that $\ell_i \leq \listmax$, and the third line follows from simplifying terms. For $Q \geq \listmax^2 \AdvTotal$, if this expectation holds, then the average cost per query is $O(\listave)$, as claimed.
\end{proof}






\section{Conclusion and Future Work}

In this paper, we have designed and analyzed an RB-based defense against algorithmic complexity attacks on hash tables. The cost of our defense grows slowly with the cost of the adversary. For example, when there is little-to-no attack, the defense has low RB cost and latency; conversely,  when the attack is significant, our defense imposes an asymptotically-higher RB cost on the adversary.

To the best of our knowledge, our defense is the first to leverage RB for defending against ACA attacks, and there are several promising questions for future work.  First, our current work addresses fixed-size hash tables, which captures settings  where there is no need to resize the hash table, or it is not desirable to do so. However, can we extend our approach to the case when legitimate system load is unpredictable {\it and} there exist adequate server-side resources to resize many times? 

Second, can we extend our approach to other data structures that employ hash functions, such as Bloom filters? Similarly, decentralized data structures other than those based on hash tables---such as skip graphs~\cite{aspnes:skip}---might also benefit from such a defense. 

Third, can we derive asymptotically matching lower bounds for the algorithm's performance? Is there a fundamental trade off between these RB cost and latency?

Finally, machine learning (ML) has become an important tool for improving the performance of algorithms~\cite{roughgarden2021,chakraborty2023bankrupting}. In our setting, it would be interesting to determine if ML predictions about whether a request is good or bad can be leveraged to improve the bounds on the cost ratio or latency.

\bibliographystyle{plain} 

\end{document}